\documentclass[a4paper]{article}

\title{On Decidability of $2$-process Affine Models} 

\usepackage{authblk}

\newcommand\ignore[1]{}

\date{}

\author[1]{Petr Kuznetsov\thanks{petr.kuznetsov@telecom-paris.fr}}
\affil[1]{LTCI, T\'el\'ecom Paris, Institut Polytechnique Paris, Paris, France}{}{}

\author[2]{Thibault Rieutord\thanks{thibault.rieutord@cea.fr}}
\affil[2]{CEA, LIST, PC 174, Gif-sur-Yvette, France}

\usepackage{amssymb}
\usepackage{subcaption}
\usepackage{amsmath}
\usepackage{epigraph}
\usepackage{amsthm}
\usepackage{a4wide}

\usepackage{url}
\usepackage{tikz}
\usepackage{color} 
\usepackage{graphics}
\usepackage{enumerate}
\usepackage{microtype}
\usepackage{hyperref}

\hypersetup{
    colorlinks=true,
    linkcolor=blue,
    filecolor=magenta,      
    urlcolor=cyan,
}

\def\I{\ensuremath{\mathcal{I}}}
\def\O{\ensuremath{\mathcal{O}}}

\def\K{\ensuremath{\mathcal{K}}}
\def\L{\ensuremath{\mathcal{L}}}

\def\C{\ensuremath{\mathcal{C}}}

\newcommand{\remove}[1]{}
\newtheorem{theorem}{Theorem}

\def\s {\mathbf{s}}
\def\t {\mathbf{t}}

\def\Chr{\operatorname{Chr}}

\def\O {\mathcal{O}}
\def\I {\mathcal{I}}

\def\IS{\textit{IS}}

\begin{document}

\maketitle

\begin{abstract}
An affine model of computation is defined as a subset of iterated immediate-snapshot runs, capturing a wide variety of shared-memory systems, such as wait-freedom, $t$-resilience, $k$-concurrency, and fair shared-memory adversaries. The question of whether a given task is solvable in a given affine model is, in general, undecidable. 

In this paper, we focus on affine models defined for a system of two processes.  We show that the task computability of $2$-process affine models is decidable and presents a complete hierarchy of the five equivalence classes of $2$-process affine models. 
\end{abstract}

\section{Introduction}

The question of whether a task is solvable in a \emph{wait-free} manner, i.e., in the asynchronous read-write shared-memory model with no restrictions on who and when can fail, is known to be undecidable for systems with more than $3$ processes~\cite{GK99-undecidable,HR97}. We can still, however, study the \emph{relative} computability of models of computation. The framework of \emph{affine models} was introduced to capture the task computability of various restrictions of the wait-free model~\cite{GKM14-podc}. 

More precisely, an \emph{affine task} $A$ on $n+1$ processes can be represented as a pure (i.e., with facets of dimension $n$) $n$-dimensional non-empty sub-complex of a finite number of iterations of the \emph{standard chromatic subdivision}, i.e., $A\subseteq \Chr^k \s, k\in \mathbb{N}, pure(A)$. Many shared-memory models such as $t$-resilience~\cite{SHG16}, $k$-set consensus~\cite{GHKR16} or the class of fair adversaries~\cite{KRH18} are characterized as affine models.
The corresponding \emph{affine model}, denoted by $A^*$, is characterized by its ability to solve tasks as follows: $A$ solves a task $(\I,\Delta,\O)$ if and only if there is a natural integer $b\in \mathbb{N}$ and a simplicial map $\delta$: $A^b (\I) \rightarrow \O$ such that $\delta$ is carried by $\Delta$, i.e., $\forall s \in I, \delta(A^b (\I)) \subseteq \Delta(s)$. A~natural question is therefore to compare relative task computability of affine models:
\begin{quote}
$A^*$ is \emph{stronger} than $B^*$, i.e., $A^*\succeq_\mathcal{A} B^*$, if all tasks solvable in $B^*$ can be solved in~$A^*$.
\end{quote}
Hence, we can state our problem as follows:
\begin{quote}
Given two affine tasks, $A$ and $B$, is the question of whether $A^*\succeq_\mathcal{A} B^*$ decidable? 
\end{quote}
Equivalently, we can study decidability of the question whether \emph{$A^*$ solves $B$}, i.e., whether $A$ solves the \emph{simplex agreement} task on $B$~\cite{BG97}.
Indeed, suppose that $A^*$ solves $B$, inductively, for any $b\in \mathbb{N}$, $A^*$ solves $B^b$. Thus, any task solvable in $B^*$ can be solved in $A^*$.

In this paper, we first present a framework for studying affine task decidability in $2$-process affine models. It allows us to provide a complete hierarchy of $2$-process affine models, including most, if not all, shared-memory models. 
We show that all $2$-process affine models fall into  five equivalence classes, each class equipped with a \emph{representative} defined as a subset of  a single iteration of the standard chromatic subdivision. The order presented in Figure~\ref{Dim1OrderSCLC} provides a complete hierarchy of their relative task computability.

An intriguing question is whether this approach can be applied to higher-dimensional systems. One approach could be to focus on models defined using  \emph{link-connected} affine tasks. 

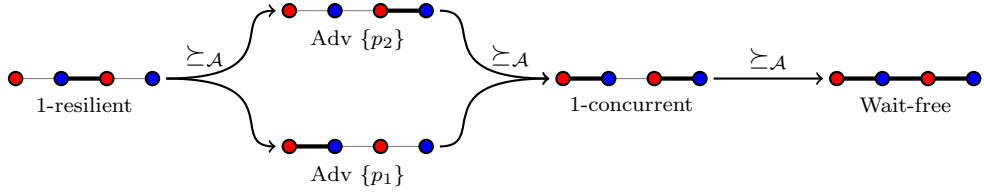
\begin{figure}
\center
		\begin{tikzpicture}[scale=1.8]
			\draw[gray] (0.000000,0.500000) -- (0.333333,0.500000);
			\draw[black,ultra thick] (0.333333,0.500000) -- (0.666667,0.500000);
			\draw[gray] (0.666667,0.500000) -- (1.000000,0.500000);
			\draw (0.000000,0.500000) node[circle,inner sep=0pt,minimum size=5pt, fill=red,draw,thick] {} ;
			\draw (0.666667,0.500000) node[circle,inner sep=0pt,minimum size=5pt, fill=red,draw,thick] {} ;
			\draw (0.333333,0.500000) node[circle,inner sep=0pt,minimum size=5pt, fill=blue,draw,thick] {} ;
			\draw (1.000000,0.500000) node[circle,inner sep=0pt,minimum size=5pt, fill=blue,draw,thick] {} ;
			\draw (0.500000,0.500000) node[label=below :{\footnotesize  $1$-resilient}] {} ;			
			\draw[thick,->]	(1.1,0.5) to[in=180,out=0,out looseness=2.5,in looseness=1] node [above,very near start] {$\succeq_\mathcal{A}$} (1.9,1);	
			\draw[thick,->]	(1.1,0.5) to[in=180,out=0,out looseness=2.5,in looseness=1] (1.9,0);

			\draw[black,ultra thick] (2.000000,0.000000) -- (2.333333,0.000000);
			\draw[gray] (2.333333,0.000000) -- (2.666667,0.000000);
			\draw[gray] (2.666667,0.000000) -- (3.000000,0.000000);
			\draw (2.000000,0.000000) node[circle,inner sep=0pt,minimum size=5pt, fill=red,draw,thick] {} ;
			\draw (2.666667,0.000000) node[circle,inner sep=0pt,minimum size=5pt, fill=red,draw,thick] {} ;
			\draw (2.333333,0.000000) node[circle,inner sep=0pt,minimum size=5pt, fill=blue,draw,thick] {} ;
			\draw (3.000000,0.000000) node[circle,inner sep=0pt,minimum size=5pt, fill=blue,draw,thick] {} ;
			\draw (2.500000,0.000000) node[label=below :{\footnotesize Adv $\{p_1\}$}] {} ;	
			
			\draw[gray] (2.000000,1.000000) -- (2.333333,1.000000);
			\draw[gray] (2.333333,1.000000) -- (2.666667,1.000000);
			\draw[black,ultra thick] (2.666667,1.000000) -- (3.000000,1.000000);
			\draw (2.000000,1.000000) node[circle,inner sep=0pt,minimum size=5pt, fill=red,draw,thick] {} ;
			\draw (2.666667,1.000000) node[circle,inner sep=0pt,minimum size=5pt, fill=red,draw,thick] {} ;
			\draw (2.333333,1.000000) node[circle,inner sep=0pt,minimum size=5pt, fill=blue,draw,thick] {} ;
			\draw (3.000000,1.000000) node[circle,inner sep=0pt,minimum size=5pt, fill=blue,draw,thick] {} ;
			\draw (2.500000,1.000000) node[label=below :{\footnotesize Adv $\{p_2\}$}] {} ;	
			
			\draw[->,thick] (3.1,1) to[in=180,out=0,in looseness=2.5,out looseness=1] node [above,very near end] {$\succeq_\mathcal{A}$} (3.9,0.5);	
			\draw[->,thick] (3.1,0) to[in=180,out=0,in looseness=2.5,out looseness=1](3.9,0.5);	
			
			\draw[black,ultra thick] (4.000000,0.500000) -- (4.333333,0.500000);
			\draw[gray] (4.333333,0.500000) -- (4.666667,0.500000);
			\draw[black,ultra thick] (4.666667,0.500000) -- (5.000000,0.500000);
			\draw (4.000000,0.500000) node[circle,inner sep=0pt,minimum size=5pt, fill=red,draw,thick] {} ;
			\draw (4.666667,0.500000) node[circle,inner sep=0pt,minimum size=5pt, fill=red,draw,thick] {} ;
			\draw (4.333333,0.500000) node[circle,inner sep=0pt,minimum size=5pt, fill=blue,draw,thick] {} ;
			\draw (5.000000,0.500000) node[circle,inner sep=0pt,minimum size=5pt, fill=blue,draw,thick] {} ;
			\draw (4.500000,0.500000) node[label=below :{\footnotesize $1$-concurrent}] {} ;
			
			\draw[->,thick]	(5.1,0.5) to node[above] {$\succeq_\mathcal{A}$} (5.9,0.5);		
			
			\draw[black,ultra thick] (6.000000,0.500000) -- (6.333333,0.500000);
			\draw[black,ultra thick] (6.333333,0.500000) -- (6.666667,0.500000);
			\draw[black,ultra thick] (6.666667,0.500000) -- (7.000000,0.500000);
			\draw (6.000000,0.500000) node[circle,inner sep=0pt,minimum size=5pt, fill=red,draw,thick] {} ;
			\draw (6.666667,0.500000) node[circle,inner sep=0pt,minimum size=5pt, fill=red,draw,thick] {} ;
			\draw (6.333333,0.500000) node[circle,inner sep=0pt,minimum size=5pt, fill=blue,draw,thick] {} ;
			\draw (7.000000,0.500000) node[circle,inner sep=0pt,minimum size=5pt, fill=blue,draw,thick] {} ;
			\draw (6.500000,0.500000) node[label=below :{\footnotesize Wait-free}] {} ;
		\end{tikzpicture}
\caption{Relations between canonical affine tasks and corresponding models.\label{Dim1OrderSCLC}}
\end{figure}

\section{Preliminaries}

Let us now recall several notions from combinatorial topology. 
For more detailed coverage of the topic, please refer to~\cite{Spanier,HKR14}.

\subparagraph*{Simplicial complex.}
A {\em simplicial complex} is a set $V$, together with an \emph{inclusion-closed}
collection $\K$ of finite non-empty subsets of $V$ such that:
\begin{enumerate}
\item For any $v \in V$, the one-element set $\{v\}$ is in $\K$;
\item If $\sigma \in \K$ and $\sigma' \subseteq \sigma$, then $\sigma' \in \K$.
\end{enumerate}

The elements of $V$ are called {\em vertices}, and the elements 
of $\K$ are called {\em simplices}. We usually drop $V$ from the notation 
and refer to the simplicial complex as $\K$ directly. Indeed, we can extract
from $\K$ the set of vertices composing it. We denote as $\mathit{Vert}(\K)$ 
the set of vertices of $\K$.
A simplicial complex $\K$ is {\em finite} if the collection $\K$ is finite. 
For simplicity, we will assume that our complexes are finite.

The \emph{dimension} of a simplex $\sigma$, denoted $\dim(\sigma)$,
is its cardinality minus one, i.e., $\#(\sigma)-1$ (the use of $|.|$ will be avoided
since it is traditionally used for the geometrical representation of a simplex). 
Any subset of a simplex $\sigma$ is also a simplex and is called a \emph{face} of $\sigma$. 
We denote as $\mathit{faces}(\sigma)$ the set containing all faces of $\sigma$. 
Given a complex $\K$ and a simplex~$\sigma\in \K$, $\sigma$ is a \emph{facet} of $\K$, 
denoted $\mathit{facet}(\sigma,\K)$, if $\sigma$ is not the face of any strictly larger simplex in $\K$.
Let $\mathit{facets}(K)= \{\sigma\in \K,\mathit{facet}(\sigma,K)\}$.
The dimension of a complex is equal to the maximal dimension of the simplices composing it.

A {\em sub-complex} of $\K$ is a subset of $\K$ that is also a simplicial complex.
A simplicial complex~$\K$ is called {\em pure} of dimension $n$ if $\K$ has no 
simplices of dimension $> n$, and every $k$-dimensional simplex of $\K$ (for $k < n$) 
is a face of an $n$-dimensional simplex of $\K$. Hence, equivalently, a simplicial complex $\K$ is 
pure of dimension $n$ if all its facets are of dimension $n$.

\subparagraph*{Chromatic complexes.}
We now turn to the chromatic complexes used in distributed computing.
Fix $n \geq 0$. The {\em standard $n$-simplex} $\s^n$ has $n+1$ vertices, 
in one-to-one correspondence with $n+1$ {\em colors} $0, 1, \dots, n$. 
A face $\t$ of $\s$ is specified by a collection of vertices from $\{0, \dots, n\}$. 
We view $\s^n$ as a complex, with its simplices being all possible faces~$\t$.

A {\em chromatic complex} is a simplicial complex $\K$ together with 
a non-collapsing simplicial map $\chi: \K \to \C$, $\C$ being a set of 
colors. See the following paragraphs for formal definitions about simplicial 
maps, but informally, it corresponds to associating  
colors to any vertex of a {\em chromatic complex} 
such that all vertices of the same simplex have distinct associated colors.
Note that therefore, $\K$ can have dimension at most $\#(\C)-1$. 
We usually drop $\chi$ from the notation and consider that vertices 
are couples $(v,c)$ where $v$ is the vertex and $c$ its 
associated color. We write $\chi(\K)$ for the union of $\chi(v)$ 
over all vertices $v \in \mathit{Vert}(\K)$. Note that if $\K' \subseteq \K$ is 
a sub-complex of a chromatic complex, it inherits a chromatic 
structure by restriction.
In particular, the standard $n$-simplex $\s^n$ is a chromatic complex, 
with $\chi$ being the identity map.

In our setting, colors correspond by default to processes identifiers. 
In this case, the set of colors of a complex is equal to $\chi(\s^n)$. 
Since most of the time the size of the system is fixed and known from the context, 
we use $\s$ to denote the standard $(\#(\Pi)-1)$-simplex. 
Note that, when colors correspond to processes identifiers, we use 
the map $\chi$ to obtain both the color of a vertex and the process 
corresponding to the identifier.  

\subparagraph*{Maps.}
Let $\K$ and $\L$ be simplicial complexes. 
A simplicial map $f: \K \to \L$ is a 
function from~$\K$ to $\L$ such that for any face $\theta$ of a 
simplex $\sigma\in \K$, then $f(\theta)$ is a face of $f(\sigma)$ in~$\L$. 
A~simplicial map is said to be non-collapsing if for any strict 
face $\theta$ of a simplex $\sigma\in \K$, then~$f(\theta)$ is a 
strict face of $f(\sigma)$.
Hence, the image of an $m$-dimensional simplex through a non-collapsing 
map is also an $m$-dimensional simplex.
Let $\K$ and $\L$ be chromatic complexes. A simplicial 
map~$f: \K \to \L$ is a \emph{color-preserving}, also called a \emph{chromatic map}, 
if for all vertices~$v \in \mathit{Vert}(\K)$, we have $\chi(v) = \chi(f(v))$. 
Note that a color-preserving map is automatically non-collapsing. 

A carrier map  $\Psi: \K \to 2^\L$ sends simplices to sub-complexes 
such that a face $\theta$ of a simplex~$\sigma\in \K$ is sent 
to a complex $\Psi(\theta)$ which is a sub-complex of $\Psi(\sigma)$.
A simplicial map $\phi$ is carried by the carrier map $\Psi$ if 
$\phi(\sigma) \in \Psi(\sigma)$ for every simplex $\sigma$ in its
domain.

\paragraph*{Continous representation.} 
We can associate a simplicial complex $\K$ with a topological space 
$|\K|$, called its geometrical realization. The geometrical 
realization is defined incrementally: first, each vertex of $\K$ 
is associated with points in $[0,1]^{\dim(\K)}$ such that vertices 
from the same simplex are associated with affinely 
independent positions; then the geometric realization of a 
simplex is equal to the convex-hull of its vertices.

Note that given a simplicial map $f: \K \to \L$, we can extend it 
linearly to obtain a continuous map $|f|:|\K| \to |\L|$: a point 
$p\in |\K|$ is a linear combination of the vertices from $\K$ and 
its image in $|f|$ is the same linear combination of the images of 
the vertices.

\paragraph*{Subdivisions}

An important notion in combinatorial topology and more specifically for 
its application to distributed computing is the notion of \emph{subdivision}.
A subdivision of a simplicial complex~$\K$, is a simplicial complex $\mathit{Sub}(\K)$ 
such that: (1) the geometrical realization of any simplex from~$\mathit{Sub}(\K)$ is 
included in the geometrical realization of a simplex from $\K$; and (2) 
the geometrical realization of a simplex in $\K$ is the union of 
geometric realizations of simplices of~$\mathit{Sub}(\K)$. Note that the subdivision 
of a subdivision of $\K$ is by definition of subdivision of~$\K$.

\subparagraph*{Standard chromatic subdivision.}
Every chromatic complex $\K$ has a {\em standard chromatic subdivision} $\Chr \K$. 
Let us first define $\Chr \s$ for the standard simplex $\s$. The vertices of~$\Chr \s$ 
are pairs $(i, \t)$, where $i \in \{0,\dots, n\}$ and $\t$ is a face of $\s$ containing $i$. 
We let $\chi(i, \t) = i$. Further,~$\Chr s$ is characterized by its $n$-simplices; 
they are the $(n+1)$-tuples $((0,\t_0), \dots, (n, \t_n))$ such that:
\begin{enumerate}[(a)]
\item For all $\t_i$ and $\t_j$, one is a face of the other;
\item If $j \in \t_i$, then $\t_j \subseteq \t_i$. 
\end{enumerate} 

Next, given a chromatic complex $\K$, we let $\Chr \K$ be the 
subdivision of $\K$ obtained by replacing each simplex in $\K$ 
with its chromatic subdivision. Thus, the vertices of $\Chr \K$ 
are pairs $(p, \sigma)$, where $p$ is a vertex of $\K$ and $\sigma$ 
is a simplex of $\K$ containing $p$. If we iterate this process $m$ times, we
obtain the $m^\mathit{th}$ chromatic subdivision, $\Chr^m \K$.

It has been shown formally by Kozlov in~\cite{Koz12} that $\Chr$ is indeed a \emph{subdivision}. In particular, 
the geometric realization of $\Chr\s$, $|\Chr\s|$, is homeomorphic to $|\s|$, 
the geometric realization of~$\s$ (i.e., the convex hull of its vertices).

If we \emph{iterate} this subdivision~$m$ times, each time 
applying $\Chr$ to all simplices, we 
obtain the~$m^{th}$ chromatic subdivision,~$\Chr^m$. 
$\Chr^m \s$ precisely captures the $m$-round IIS model,~$\IS^m$~\cite{BG97,HS99}.

\subparagraph*{Carriers.}
Given a complex $\K$ and a subdivision of it, $\mathit{Sub}(\K)$, 
the carrier of a simplex $\sigma\in \mathit{Sub}(\K)$ in $\K$, 
$\mathit{carrier}(\sigma,\K)$, is the smallest simplex $\rho\in \K$ 
such that the geometric realization of~$\sigma$,~$|\sigma|$, is contained in $|\rho|$: 
$|\sigma|\subseteq|\rho|$. 
The carrier of a vertex~$(p,\sigma)\in\Chr \s$ is $\sigma$. 
In the matching $\IS$ task, the carrier corresponds to the snapshot returned by~$p$, 
i.e., the set of processes \emph{seen} by $p$. The carrier of a simplex 
$\rho \in \Chr \K$ is just the union (or, due to inclusion, the maximum) 
of the carriers of vertices in $\rho$. 
Given a simplex $\sigma\in\Chr^2\s$, $\mathit{carrier}(\sigma,\s)$ is equal to 
$\mathit{carrier}(\mathit{carrier}(\sigma,\Chr\s),\s)$. $\mathit{carrier}(\sigma,\Chr\s)$
corresponds to the set of all snapshots seen by processes in $\chi(\sigma)$.
Hence, $\mathit{carrier}(\sigma,\s)$ corresponds to the union of all these snapshots. 
Intuitively, it results in the set of all processes \emph{seen} by processes 
in $\chi(\sigma)$ through the two successive immediate snapshots instances.

\subparagraph*{Simplex agreement task.}
In the \emph{simplex agreement task}, processes start on vertices of some complex~$\K$ 
forming a simplex $\sigma\in \K$, and they must output vertices of some subdivision of~$\K$, 
$\mathit{Sub}(\K)$, so that outputs constitute a simplex $\rho$ of $\mathit{Sub}(\K)$ respecting carrier 
inclusion,~i.e.,~$\mathit{carrier}(\rho,\K)\subseteq\sigma$.

Such tasks are primordial for the proof of the asynchronous computability 
theorem (ACT)~\cite{HS99}. Indeed, given a simplicial map from a subdivision of 
the task input complex solving the task, processes must manage first to solve 
the simplex agreement task on the given subdivision to be able to apply 
the task solution provided by the simplicial map. In the original version of the ACT, 
a map could be given for an arbitrary subdivision. But using the equivalence with the 
IIS model, it was shown that we could  improve the result 
by considering only iterations of the standard chromatic subdivision~\cite{BG97}.

\subparagraph*{Convergence algorithm.}
The next theorem is a corollary of work by Borowsky and Gafni~\cite{BG97} (and described in detail by Saraph et al.~\cite{SHG18}), which turns continuous maps into simplicial maps.

\begin{theorem}[convergence algorithm]
Let $\mathcal{I}$ and $\mathcal{O}$ be chromatic complexes,
$\Gamma:~\mathcal{I}~\to~2^\mathcal{O}$ a carrier map such that $\Gamma(\sigma)$ is
link-connected for each $\sigma \in \mathcal{I}$,
and $f:~|\mathcal{I}|~\to~|\mathcal{O}|$ a continuous map carried by $\Gamma$.
Then there exists a chromatic, carrier-preserving simplicial map
$\phi:~\Chr^N(\mathcal{I})~\to~\mathcal{O}$, for some sufficiently large $N$,
also carried by $\Gamma$.\label{thm:conv}
\end{theorem}

\paragraph*{Affine tasks}

An \emph{affine task} is a generalization of the simplex agreement
task, where the output complex is a pure non-empty sub-complex 
of some finite number of iterations of the standard chromatic subdivision, 
$\Chr^{\ell}\s$. 
Formally, let $\L$ be a pure non-empty sub-complex of $\Chr^{\ell}\s^n$ 
of dimension~$n$ for some~$\ell\in \mathbb{N}$. The affine task associated to $\L$ 
is then defined as $(\s^n,\L,\Delta)$, where, for every face $\sigma \subseteq \s^n$,
$\Delta(\sigma) = \L \cap \Chr^{\ell}(\sigma)$. 
Hence, processes start on vertices of their color in~$\s^n$ and must eventually output 
vertices of $\L$ of their color such that the set of outputs forms 
a simplex in $\L$ with a carrier equal to the set of observed processes. 
Note that $\L \cap \Chr^{\ell}(\t)$ can be empty, in which case processes
are not allowed to output views containing only the inputs from processes in $\chi(\t)$.
Intuitively, this is used to guarantee that the model can provide sufficiently 
large participation to be able to solve the affine task.

Note that, since an affine task is characterized by its output complex, 
with a slight abuse of notation, we use~$\L$ for both the affine task 
$(\s,\L,\Delta)$ and its output complex.

\subparagraph*{Affine model.}
It can be noted that an affine task can also be seen as an operator 
on any pure simplicial complexes of the same dimension. Indeed, given 
a pure simplicial complex~$\K$, we can construct the simplicial complex
$\L(\K)$ where each facet of $\K$ is replaced by an occurrence of $\L$. 
In particular, since the operation maintains purity, we can iterate this 
operation by recursively replacing simplices by an occurrence of the affine task.
 
By running~$m$ iterations of this operation on $\s$, we obtain $\L^m(\s)$, a 
sub-complex of $\Chr^{\ell m}\s$. 
The affine model $\L^*$ corresponding to the affine task $\L$ 
is obtained by iterating infinitely often the affine task. 
The $m^{th}$ iteration of an affine task corresponds to 
a subset of $\IS^{~\ell m}$ runs (as each of the $m$ iterations includes $\ell$ $\IS$ rounds). 
Hence the affine model $\L^*$ corresponds to the set of infinite runs 
of the IIS model where every prefix restricted to a multiple of 
$\ell$ $\IS$ rounds belongs to the subset of $\IS^{~\ell m}$ 
runs associated with $\L^m$. 

Note that, by construction, affine models are compact. 
Indeed, they are defined through a ``safety'' property 
on the set of IIS runs: if all prefixes of an IIS run satisfy 
the model conditions then the infinite run belongs to the model.

\section{Computing equivalence classes}

We  start with identifying simple equivalence classes on $2$-process affine tasks via a simple predicate on a set of properties. 
We start by defining a partition on affine tasks and then show that tasks in the same class are equivalents, we can then select representative of each classes.

\subparagraph*{Property selection.}  The power a $2$-process system heavily relies on the properties of \emph{solo executions}, i.e., the endpoints of the corresponding affine task.
Assuming a fixed input state, there is only one such an endpoint $v_0$ of process $p_0$ and one endpoint $v_1$ of process $p_1$.

More formally, in a subdivision of the standard simplex composed of vertices $p_1$ and $p_2$, there are a single vertex with a carrier equal to $p_1$ or $p_2$ relatively to the standard simplex. Indeed, a subdivision replaces simplices with the subdivision of the face with the same carrier. Therefore, as endpoints are of dimension $0$, they are replaced with a single vertex sharing the same carrier. These are the vertices we call $v_0$ and $v_1$ in any given affine task.

We can then identify the following classes of $2$-process affine tasks:
\begin{enumerate}
\item There is a path from $v_0$ to $v_1$ ($v_0$ and $v_1$ are simply connected).
\item No path between $v_0$ to $v_1$, but $v_0$ and $v_1$ belong to the task.
\item No path between $v_0$ to $v_1$, but only $v_0$ belong to the task.
\item No path between $v_0$ to $v_1$, but only $v_1$ belong to the task.
\item No path between $v_0$ to $v_1$, and neither $v_0$ nor $v_1$ belong to the task.
\end{enumerate}

Let us now show that these classes are disjoint and form a partition of all $2$-process affine tasks. We also show afterward that computing the class of a $2$-process affine task is decidable.

\begin{theorem}
The set of $2$-processes affine tasks classes form a disjoint partition of the set of $2$-process affine tasks. 
\end{theorem}

\begin{proof}
Classes are defined based on $3$ properties: (1) $P_1$: $v_0$ and $v_1$ are simply connected; (2) $P_2$: $v_0$ belong to the task; and (3) $P_3$: $v_1$ belong to the task. We can reformulate the classes according to these properties as follows:
\begin{enumerate}
\item tasks satisfying $F_1=P_1$.
\item tasks satisfying $F_2=\neg P_1\wedge P_2 \wedge P_3$.
\item tasks satisfying $F_3=\neg P_1\wedge P_2 \wedge \neg P_3$.
\item tasks satisfying $F_4=\neg P_1\wedge \neg P_2 \wedge P_3$.
\item tasks satisfying $F_5=\neg P_1\wedge \neg P_2 \wedge \neg P_3$.
\end{enumerate}
To check that it forms a partition, we need to check that (1) any two formulas cannot be both satisfied at the same time, and (2) there is always a satisfied formula given any state of the properties. Hence that $\forall i,j\in \{1,\dots,5\}, i\neq j, F_i\wedge F_j=\bot$ and that $F_1\vee F_2\vee F_3\vee F_4 \vee F_5=\top$. For the former case, we leave the reader to check all cases thoroughly, and we just note that among two formulas, a property is flipped and must be satisfied in one and not satisfied in the other. For the latter, we simply point out that $F_1$ corresponds to $P_1$ and all others correspond to a conjunction of $\neg P_1$ with one of four possibilities of $P_2$ and $P_3$ satisfiability.
\end{proof}

\begin{theorem}
Computing the class of a $2$-process affine task is decidable.
\end{theorem}

\begin{proof}
For this result, we only need to check that computing whether the properties $P_1$, $P_2$ and $P_3$ are satisfied is a decidable question. Identifying vertices $v_0$ and $v_1$ is trivial as computing the carrier of a vertex is part of the subdivision definition, it corresponds to a simple inclusion test. For property $P_1$, we only need to execute a graph search to check whether $v_0$ and $v_1$ are simply-connected or not.
\end{proof}

\subparagraph*{Tasks in the same class are equivalent.} Let us show that for any couple of tasks $A$ and~$B$ in the same class, $A^*$ solves $B$. This way we show that all tasks in the same class are equivalent to each other:

\begin{theorem}\label{thm:equiv}
Tasks in the same class are equivalent to each other.
\end{theorem}

\begin{proof}
Let us first consider the case of the first class where a path exists between $v_0$ and~$v_1$. Consider two affine tasks $A$ and $B$ belonging to this class. The existence of a path translates to the geometric realization of both tasks. Note that the carrier of elements of these path are equal to $\{p_0,p_1\}$ except for the endpoints $v_0$ and $v_1$. Therefore we can map the continuous path from $|A|$ to the path of $|B|$ in a carrier-preserving manner. Lastly, affine tasks of dimension $1$ are always link-connected as the link is not-empty for pure complexes. Thus, we can apply the convergence algorithm to obtain a simplicial map from sufficiently many subdivisions of $A$ to $B$. In particular, $A$ is a subset of a subdivision; thus, we obtain a map from some iteration of $A$ to $B$ by a carrier and color-preserving map. It completes the proof that $A^*$ solves the simplex agreement task on $B$. Hence, all tasks in the first class are equivalent to each other.

Let us now look at the remaining cases altogether. Consider two affine tasks $A$ and $B$ in the same class. Let us split the simplices of $A$ in each task by connected components. Note that, since there is no paths from $v_0$ to $v_1$, there is no connected component including both $v_0$ and $v_1$. Depending on whether $v_0$, resp. $v_1$, belong to the task or not, there is a connected component including $v_0$, resp. $v_1$. Such a component can be mapped to the facet of $B$ containing $v_0$, resp. $v_1$. Indeed, if there is such a connected component for $A$, then as $B$ belongs to the same class of affine task, it must possess such a facet. We are left with connected components containing neither $v_0$ nor $v_1$, hence with simplices with a carrier equal to $\{p_0,p_1\}$ that can be mapped to the simplices of any facet of $B$ in a carrier and color-preserving way. Hence, $A^*$ solves $B$, and we have shown that all tasks in this class are equivalent to each other for any class.\end{proof}

\subparagraph*{Canonical tasks.}
In each equivalence class, we can then select a characterizing representative, which we call a \emph{canonical task}. 
We will show that the partial order on these canonical tasks (provided in Figure~\ref{Dim1OrderSCLC}) captures the relative power of the equivalence classes. Figure~\ref{Dim1List} provides the set of canonical affine tasks. The affine model of a canonical task is also called canonical.

\section{Comparing equivalence classes}

To show that our selected partial order corresponds to affine models relative task computability power, we need to show that: (1) iterations of affine tasks cannot increase the associated class; (2) carrier-preserving simplicial maps can only send tasks to tasks in smaller or equal classes; and (3) canonical affine models follow this order.
It is easy to check that (1) and (2) implies that equivalent affine models belong to the same class. As all models in a class are equivalent, comparing canonical tasks (3) is, hence, sufficient to compare all models. Moreover, (1) and (2) also imply that models in a class cannot solve tasks in higher classes; consequently, (3) reduces to showing that a higher canonical model is stronger than a smaller~one.

\begin{figure}
\captionsetup[subfigure]{justification=centering}
  \begin{minipage}[b]{.19\linewidth}
    \centering
		\begin{tikzpicture}[scale=2]
			\draw[gray] (0.000000,0.000000) -- (0.333333,0.000000);
			\draw[black,ultra thick] (0.333333,0.000000) -- (0.666667,0.000000);
			\draw[gray] (0.666667,0.000000) -- (1.000000,0.000000);
			\draw (0.000000,0.000000) node[circle,inner sep=0pt,minimum size=5pt, fill=red,draw,thick] {} ;
			\draw (0.666667,0.000000) node[circle,inner sep=0pt,minimum size=5pt, fill=red,draw,thick] {} ;
			\draw (0.333333,0.000000) node[circle,inner sep=0pt,minimum size=5pt, fill=blue,draw,thick] {} ;
			
			\draw (1.000000,0.000000) node[circle,inner sep=0pt,minimum size=5pt, fill=blue,draw,thick] {} ;
		\end{tikzpicture}
    \subcaption{No $v_1$, no $v_2$}\label{fig:S000}
  \end{minipage}
  \hspace{0.001\linewidth}
  \begin{minipage}[b]{.19\linewidth}
    \centering
		\begin{tikzpicture}[scale=2]
			\draw[black,ultra thick] (0.000000,0.000000) -- (0.333333,0.000000);
			\draw[gray] (0.333333,0.000000) -- (0.666667,0.000000);
			\draw[gray] (0.666667,0.000000) -- (1.000000,0.000000);
			\draw (0.000000,0.000000) node[circle,inner sep=0pt,minimum size=5pt, fill=red,draw,thick] {} ;
			\draw (0.666667,0.000000) node[circle,inner sep=0pt,minimum size=5pt, fill=red,draw,thick] {} ;
			\draw (0.333333,0.000000) node[circle,inner sep=0pt,minimum size=5pt, fill=blue,draw,thick] {} ;
			
			\draw (1.000000,0.000000) node[circle,inner sep=0pt,minimum size=5pt, fill=blue,draw,thick] {} ;
		\end{tikzpicture}
    \subcaption{Only $v_1$}\label{fig:S100}
  \end{minipage}
  \hspace{0.001\linewidth}
  \begin{minipage}[b]{.19\linewidth}
    \centering
		\begin{tikzpicture}[scale=2]
			\draw[gray] (0.000000,0.000000) -- (0.333333,0.000000);
			\draw[gray] (0.333333,0.000000) -- (0.666667,0.000000);
			\draw[black,ultra thick] (0.666667,0.000000) -- (1.000000,0.000000);
			\draw (0.000000,0.000000) node[circle,inner sep=0pt,minimum size=5pt, fill=red,draw,thick] {} ;
			\draw (0.666667,0.000000) node[circle,inner sep=0pt,minimum size=5pt, fill=red,draw,thick] {} ;
			\draw (0.333333,0.000000) node[circle,inner sep=0pt,minimum size=5pt, fill=blue,draw,thick] {} ;
			
			\draw (1.000000,0.000000) node[circle,inner sep=0pt,minimum size=5pt, fill=blue,draw,thick] {} ;
		\end{tikzpicture}
    \subcaption{Only $v_2$}\label{fig:S010}
  \end{minipage}
  \hspace{0.001\linewidth}
  \begin{minipage}[b]{.19\linewidth}
    \centering
		\begin{tikzpicture}[scale=2]
			\draw[black,ultra thick] (0.000000,0.000000) -- (0.333333,0.000000);
			\draw[gray] (0.333333,0.000000) -- (0.666667,0.000000);
			\draw[black,ultra thick] (0.666667,0.000000) -- (1.000000,0.000000);
			\draw (0.000000,0.000000) node[circle,inner sep=0pt,minimum size=5pt, fill=red,draw,thick] {} ;
			\draw (0.666667,0.000000) node[circle,inner sep=0pt,minimum size=5pt, fill=red,draw,thick] {} ;
			\draw (0.333333,0.000000) node[circle,inner sep=0pt,minimum size=5pt, fill=blue,draw,thick] {} ;
			
			\draw (1.000000,0.000000) node[circle,inner sep=0pt,minimum size=5pt, fill=blue,draw,thick] {} ;
		\end{tikzpicture}
    \subcaption{$v_1$ and $v_2$}\label{fig:S110}
  \end{minipage}
  \hspace{0.001\linewidth}
   \begin{minipage}[b]{.19\linewidth}
    \centering
		\begin{tikzpicture}[scale=2]
			\draw[black,ultra thick] (0.000000,0.000000) -- (0.333333,0.000000);
			\draw[black,ultra thick] (0.333333,0.000000) -- (0.666667,0.000000);
			\draw[black,ultra thick] (0.666667,0.000000) -- (1.000000,0.000000);
			\draw (0.000000,0.000000) node[circle,inner sep=0pt,minimum size=5pt, fill=red,draw,thick] {} ;
			\draw (0.666667,0.000000) node[circle,inner sep=0pt,minimum size=5pt, fill=red,draw,thick] {} ;
			\draw (0.333333,0.000000) node[circle,inner sep=0pt,minimum size=5pt, fill=blue,draw,thick] {} ;
			
			\draw (1.000000,0.000000) node[circle,inner sep=0pt,minimum size=5pt, fill=blue,draw,thick] {} ;
		\end{tikzpicture}
    \subcaption{Connected}\label{fig:S111}
  \end{minipage}

\caption{Representative affine tasks for distinct values of $S$.\label{Dim1List}}
\end{figure}
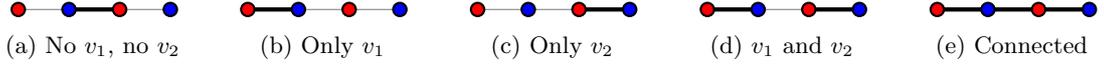

\subparagraph*{Decreasing through iterations.} Let us first show that iterating an affine task can only make it belong to a weaker class. In practice, we show that the task remains in the same class. For this we show that the properties selected for the class selection are stable under iterations, that is:

\begin{theorem}\label{thm:iter}
For any $2$-process affine task $A$ and any $k\in\mathbb{N}$, $\mathit{class}(A^k)\leq\mathit{class}(A)$.
\end{theorem}

\begin{proof}
Let us show that properties $P_1$, $P_2$ and $P_3$ are stable under iterations. Let us start with $P_2$ and $P_3$. When iterating, a simplex is replaced with the face corresponding to its colors. Hence, a vertex with carrier equal to $p_0$ or $p_1$ is replaced with a carrier of the same carrier. Therefore, $P_2$ is satisfied for an iteration if and only if it is satisfied with the original task. The same holds for $P_3$. 

Let us now look at the property $P_1$. If $P_1$ is satisfied for a given affine task, then there exist a sequence of vertices, such that each couple of successive vertices form a vertex in the task and with the first vertex with a carrier equal to $p_0$ and the last with a carrier equal to $p_1$. When iterating, each simplex of consecutive vertices is replaced by the affine task, including a path between the end vertices. Therefore, by a trivial induction, we obtain a path between the end vertices of the original sequence. It corresponds to the vertices with carrier equal to $p_0$ and $p_1$, respectively. Therefore the existence of a path is stable under iterations. 

We have shown that properties $P_1$, $P_2$ and $P_3$ are stable under iterations. The negation is not necessary, indeed, a task satisfying $P_1$ is weaker that a task not satisfying $P_1$. Moreover, when $P_1$ is not satisfied, a task satisfying more of the properties $P_2$ or $P_3$ is weaker.
\end{proof}

\subparagraph*{Decreasing through simplicial maps.} 
Our goal is to show that the affine model relative order is identical to the class order that we have defined. We start by showing that a task in some class can only solve a task in a weaker or equal class. We have shown that it is true for iterations, but solvability also considers the projection through a carrier and color-preserving map. Hence, let us show that such maps can only reduce the class of the task: 

\begin{theorem}\label{thm:map}
For any $2$-process affine task $A$ and any carrier and color-preserving map $\delta$, $\mathit{class}(\delta(A))\leq\mathit{class}(A)$.
\end{theorem}

\begin{proof}
Let us start with the more straightforward case: properties $P_2$ and $P_3$ are conserved through a carrier-preserving simplicial map. It is a direct result of the carrier-preserving notion, as the image of the vertex with carrier $p_1$ or $p_2$ must be mapped to a vertex with a smaller or equal carrier, hence in this case, the same carrier. Therefore, if an affine task satisfies the property $P_2$ or the property $P_3$, then its image does too.

Now, let us look at the property $P_1$ and show that it is also stable under carrier-preserving simplicial mapping. It results from the fact that the image of a path is a path, and that the endpoints with a carrier of size $1$ must be mapped to the vertices with the same carrier. Therefore, as for the proof of Theorem~\ref{thm:iter}, the conservation of property $P_1$, $P_2$ and $P_3$ implies that the image must be a task with a smaller or equal class.
\end{proof}

\subparagraph*{Comparing canonical models.} 
 If neither $v_0$ nor $v_1$ belongs to the task, we can map all facets to any other task facet. Hence, this canonical task is stronger than all.
For other canonical tasks, the order follows a direct task inclusion, which implies the solvability of canonical tasks in smaller classes (the solution being the identity map).

Let us show that all our partial results combine properly to show that the class order corresponds to the affine model relative computability order.
\begin{theorem}
Given two affine tasks $A$ and $B$, $A^*\preceq_\mathcal{A} B^*$ if and only if $\mathit{class}(A)\leq \mathit{class}(B)$.
\end{theorem}

\begin{proof}
Assume first that we have  $A^*\preceq_\mathcal{A} B^*$, hence we have that $B^*$ solves $A$. This implies that there exists $k\in\mathbb{N}$ and a color and carrier-preserving map $\delta$ such that $\delta(B^k)=A$. According to theorem~\ref{thm:map}, we obtain that $\mathit{class}(A)\leq\mathit{class}(B^k)$, and according to theorem~\ref{thm:iter}, that, hence, $\mathit{class}(A)\leq\mathit{class}(B)$. 

Now let us assume that we have $\mathit{class}(A)\leq\mathit{class}(B)$. Therefore, the same holds for the canonical representatives $\overline{A}$ and $\overline{B}$. But as we have shown, $\mathit{class}(\overline{A})\leq\mathit{class}(\overline{B})$ implies that~$\overline{B}^*$ solves $\overline{A}$. But as task in the same class are equivalent (Theorem~\ref{thm:equiv}), this implies that $A^*\preceq_\mathcal{A} B^*$.
\end{proof}

We have shown that determining the relative computability of affine tasks is equivalent to computing their class. We have shown that determining to which class an affine task belongs is decidable. Therefore, determining their relative computability power is decidable as well.

\def\noopsort#1{} \def\No{\kern-.25em\lower.2ex\hbox{\char'27}}
  \def\no#1{\relax} \def\http#1{{\\{\small\tt
  http://www-litp.ibp.fr:80/{$\sim$}#1}}}

\end{document}